\documentclass[a4paper,11pt]{amsart}

\usepackage{amsthm}
\usepackage{fullpage}
\usepackage{amsmath,amsfonts,amssymb} 
\usepackage{tabularx}
\RequirePackage{slashed}

\numberwithin{equation}{section}

\def\A{\mathcal{A}}
\DeclareMathOperator{\ad}{ad}
\def\bA{\mathbb{A}}
\DeclareMathOperator{\Ad}{Ad}

\def\B{\mathcal{B}}
\def\bar{\overline}

\def\C{\mathbb{C}}

\def\cO{\mathcal{O}}

\DeclareMathOperator{\diag}{diag}

\def\dirac{\slashed{\partial}}

\def\End{\textup{End}}

\def\H{\mathcal{H}}

\def\id{\mathrm{id}}

\def\su{\mathfrak{su}}

\DeclareMathOperator{\tr}{Tr}
\def\tilde{\widetilde}

\def\volelt{\sqrt{g}\,\mathrm{d}^4x}

\def\antisquark{\overline{\widetilde{q}}}
\def\squark{\widetilde{q}}

\newcommand{\inpr}[2]{\langle #1 , #2 \rangle}
\newcommand{\rinpr}[2]{(#1,#2)}

\theoremstyle{plain}
\newtheorem{thm}{Theorem}[section]
\newtheorem{lem}[thm]{Lemma}
\newtheorem{prop}[thm]{Proposition}

\newtheorem{defin}[thm]{Definition}
\newtheorem{exmpl}[thm]{Example}
\newtheorem{cor}{Corrolary}[section]

\newbox\ncintdbox \newbox\ncinttbox
\setbox0=\hbox{$-$}
\setbox2=\hbox{$\displaystyle\int$}
\setbox\ncintdbox=\hbox{\rlap{\hbox
    to \wd2{\kern-.1em\box2\relax\hfil}}\box0\kern.1em}
\setbox0=\hbox{$\vcenter{\hrule width 4pt}$}
\setbox2=\hbox{$\textstyle\int$}
\setbox\ncinttbox=\hbox{\rlap{\hbox
    to \wd2{\kern-.05em\box2\relax\hfil}}\box0\kern.1em}

\title{Supersymmetric QCD and noncommutative geometry}
\author{Thijs van den Broek and Walter D. van Suijlekom}
\address{
Institute for Mathematics, Astrophysics and Particle Physics, 
Faculty of Science, Radboud University Nijmegen,
Heyendaalseweg 135, 6525 ED Nijmegen, The Netherlands
}
\date{18 March 2010}
\email{T.vandenBroek@science.ru.nl; waltervs@math.ru.nl}

\begin{document}
\begin{abstract}
We derive supersymmetric quantum chromodynamics from a noncommutative manifold, using the spectral action principle of Chamseddine and Connes. After a review of the Einstein--Yang--Mills system in noncommutative geometry, we establish in full detail that it possesses supersymmetry. This noncommutative model is then extended to give a theory of quarks, squarks, gluons and gluinos by constructing a suitable noncommutative spin manifold (i.e. a spectral triple). The particles are found at their natural place in a spectral triple: the quarks and gluinos as fermions in the Hilbert space, the gluons and squarks as bosons as the inner fluctuations of a (generalized) Dirac operator by the algebra of matrix-valued functions on a manifold. The spectral action principle applied to this spectral triple gives the Lagrangian of supersymmetric QCD, including soft supersymmetry breaking mass terms for the squarks. We find that these results are in good agreement with the physics literature.

\end{abstract}
\maketitle


\section{Introduction}
Over the last few decades, noncommutative geometry \cite{C94} has proven to be very successful in deriving models in high-energy physics from geometrical principles. This started with the particle theories studied by Connes and Lott from a noncommutative perspective \cite{CL91}, culminating in the work of Chamseddine and Connes \cite{CC96,CC97}. Therein, the full Standard Model of high-energy physics ---including the Higgs field--- was derived from a noncommutative manifold, through the so-called spectral action principle. This principle puts gauge theories such as the Standard Model on the same geometrical footing as Einstein's general theory of relativity
by deriving a Lagrangian from a noncommutative spacetime. For more details, see eg. Section \ref{ch:prel} below. More recently, in \cite{CCM07} (see also \cite{CM07}) this noncommutative model was enhanced to also include massive neutrinos while solving a technical issue (i.e. `fermion doubling') at the same time.

Ever since the early models introduced by Connes and Lott, there has been interest in the connection between noncommutative geometry and supersymmetry. An early instance of this subject is found in \cite{HT91a,HT91b}, and also \cite{Cha94,KW96}. However, this was all before the elegant spectral action principle was introduced, in particular the last article needed to turn the noncommutative algebra of coordinates turned into a superalgebra. Throughout the present paper, the algebra of noncommutative coordinates are $M_N(\C)$-valued functions on spacetime, that is, $\A= C^\infty(M, M_N(\C))$ (possibly with $N=3$). In the paradigm of noncommutative geometry, the gauge group consists of special unitary elements in this algebra; in this case $SU(\A) = C^\infty(M,SU(N))$. The supersymmetric gauge theories we will derive thus have $SU(N)$ as a gauge group. Then, as is intended, the supersymmetry will manifest itself as a transformation  between bosonic and fermionic degrees of freedom; this was suggested in \cite{CC97}. The natural place for the fermionic degrees of freedom is in the Hilbert space of spinors. As we will see below, the bosonic degrees of freedom are generated naturally by a generalized Dirac operator; this is very similar to the origin of the Higgs boson through the finite Dirac operator in the noncommutative description of the Standard Model \cite{CCM07}.

This article is organized as follows. We start by giving a short overview of the spectral action in noncommutative geometry, since it is the main technique exploited in this article. In Section \ref{ch:susy_in_EYM} we demonstrate that the Einstein--Yang--Mills system as derived from a noncommutative manifold ---which we recall in Section \ref{ch:eym}--- is actually supersymmetric. More precisely, it is $\mathcal{N}=1$ supersymmetric $SU(N)$ Yang--Mills theory, minimally coupled to gravity. 

Section \ref{ch:superQCD} forms the main part of this article, we define a noncommutative manifold on which the spectral action gives the Lagrangian of supersymmetric quantum chromodynamics (QCD). Besides a quark and a gluon we recognize their superpartners: the squark and the gluino. The squark appears naturally as the finite part of the inner fluctuations of the noncommutative manifold, besides the gluons as the continuous part. We discuss the several terms that appear in the spectral action and find that they coincide with the usual dynamics and interaction terms between gluons, gluinos, quarks and squarks that appear in the physics literature.

\section{Preliminaries}
\label{ch:prel}

In \cite{CC97}, Chamseddine and Connes introduced the spectral action principle as a powerful device to derive (potentially physical) Lagrangians from a noncommutative spin manifold. For convenience, we will start by quickly recalling their setup and approach. 

The basic device in noncommutative geometry \cite{C94} is a {\it spectral triple} $(\A,\H,D)$ consisting of a $*$-algebra $\A$ of bounded operators in a Hilbert space $\H$, and an unbounded self-adjoint operator $D$ in $\H$, such that
\begin{enumerate}
\item The commutator $[D,a]$ is a bounded operator;
\item The resolvent $(i+D)^{-1}$ of $D$ is a compact operator.
\end{enumerate}
One may further enrich this set of data by a self-adjoint operator $\gamma$ on $\H$ that commutes with all elements in $\A$ and is such that $\gamma^2=1$ ({\it grading}), and an anti-unitary operator $J$ on $\H$ ({\it reality}) such that the following hold
\begin{align}
 [[D, a], J b J^{-1}] = 0, \qquad [a, J b J^{-1}]=0; \qquad\forall a,b \in \A
\label{eq:order_one_condition}
\end{align}
These conditions are called the {\it first-order condition} and the {\it commutant property}, respectively. 
The following $\pm$-signs for the commutation relations between $J$, $\gamma$ and $D$,
\begin{table}[h!]
\begin{tabularx}{\textwidth}{X cccc X}
  \hline 
  & KO-dimension & $J^2 = \epsilon$ & $JD = \epsilon' DJ$ & $J\gamma = \epsilon''\gamma J$ & 
\\
         \hline
        & 0 & + & + & +&\\ 
        & 2 & $-$ & + & $-$ &\\
        & 4 & $-$ & + & + &\\
        & 6 & + & + & $-$&\\
        \hline
\end{tabularx}
\label{tb:ko-dimensions}
\end{table} 

\noindent determine the so-called KO-dimension of the real spectral triple. The notion of a real spectral triple generalizes Riemannian spin geometry to the noncommutative world. In fact, there exists a reconstruction theorem \cite{C96,C08} which states that if the algebra $\A$ in $(\A,\H,D)$ is commutative, then the spectral triple is of the form $(C^\infty(M), L^2(M,S), \dirac)$, canonically associated to a Riemannian spin manifold $M$. Here $S \to M$ is a spinor bundle and $\dirac$ is the corresponding Dirac operator.

\subsection{Inner fluctuations}
Rather than isomorphisms of algebras, a natural notion of equivalence for noncommutative ($C^*$-)algebras is Morita equivalence \cite{Rie74}. Given a spectral triple $(\A,\H,D)$ and an algebra $\B$ that is Morita equivalent to $\A$, one can define \cite{C96} a spectral triple $(\B,\H',D')$ on $\B$. Interestingly, upon taking $\B$ to be $\A$, this leads to a whole family of spectral triples $(\A, \H, D_A)$ where
$ D_A := D + A$ with $A \in \Omega_D^1(\A)$ self-adjoint with
\begin{align}
 \Omega^1(\A)  :=\big\{ \sum_i a_i[D, b_i]: a_i, b_i \in \A \big\}\label{eq:innerfluctuationform}.
\end{align}
The bounded operators $A$ are generally referred to as the \emph{inner fluctuations} of $D$ and may be interpreted as gauge fields. 

\medskip

When considering a real spectral triple $(\A, \H, D; J)$, we have the additional restriction that the real structure $J'$ of the spectral triple $(\A, \H', D'; J')$ on the Morita equivalent algebra should be compatible with the relation $J'D' = \epsilon'D'J'$. Upon taking $\B$ to be $\A$ again in such a case, the resulting spectral triple is of the form
$
  (\A, \H, D_A; J)
$, but now with
\begin{align}
D_A := A + \epsilon' JAJ^*.\label{eq:definnerf}
\end{align}
Note that in the commutative case these inner fluctuations vanish.
The \emph{gauge group} is defined to be $U(\A) := \{ u \in \A : uu^* = u^* u =1 \}$. It acts on elements $\psi$ in the Hilbert space via $\psi \mapsto u JuJ^{-1} \psi$. This induces an action on $D_A$ as $D_A \mapsto u D_A u^*$. Consequently, the inner fluctuations transform as
  \begin{align}
A \mapsto A^u := uAu^* + u[D, u^*]\label{eq:A_gauge_trans}.
  \end{align}
In the presence of a determinant on $\A$, we can restrict $U(\A)$ to $SU(\A)$ for which in addition the determinant is equal to the identity.

\subsection{The spectral action}
The above suggests that a (real) spectral triple defines a gauge theory, with the gauge fields arising as the inner fluctuations of the Dirac operator and the gauge group is given by the unitary elements in the algebra. It is thus natural to seek for gauge invariant functionals of $A \in \Omega^1_D(\A)$ and the so-called spectral action \cite{CC97} is the most natural. 

Let $(\A,\H,D; J)$ be a real spectral triple. Given the above operator $D_A$, a \emph{cut-off scale} $\Lambda$ and some positive, even function $f$ one can define (cf.~\cite{C96}, \cite{CC97}) the gauge invariant \emph{spectral action}:
\begin{align}
S_b[A]&:= \tr f(D_A/\Lambda).\label{eq:spectral_action}
\intertext{The cut-off parameter $\Lambda$ is used to obtain an asymptotic series for the spectral action; the physically relevant terms then appear with a positive power of $\Lambda$ as a coefficient.
Besides this bosonic action, one can define a fermionic action in terms of $\psi \in \H$ and $A \in \Omega^1_D(\A)$:}
S_f[A,\psi]&:= \langle \psi, D_A \psi \rangle \label{eq:ferm_action}
\end{align}
It was shown in \cite{CC97} that for a suitable choice of a spectral triple the spectral action equals the full Standard Model Lagrangian, including the Higgs boson. More recently, in \cite{CCM07} (see also \cite{CM07}) this was enhanced to also include massive neutrinos while solving a technical issue (i.e. `fermion doubling' as pointed out in \cite{LMMS97}) at the same time. We will not further go into details but refer to the mentioned literature instead.

\medskip
For convenience, we end this section by recalling some results on heat kernel expansions and Seeley--DeWitt coefficients, these will be useful later on; for more details we refer to \cite{Gil84}. If $V$ is a vector bundle on a compact Riemannian manifold $(M, g)$ and if $P:C^{\infty}(V)\to C^{\infty}(V)$ is a second-order elliptic differential operator of the form
\begin{equation}
 P = - \big(g^{\mu\nu}\partial_{\mu}\partial_{\nu} + K^{\mu}\partial_{\mu} + L)\label{eq:elliptic} 
\end{equation}
with $K^{\mu}, L \in \Gamma(\End(V))$, then there exist a unique connection $\nabla$ and an endomorphism $E$ on $V$ such that
\begin{equation}
  P = \nabla\nabla^* - E\label{eq:gilkey-operator}.
\end{equation}
Explicitly, we write locally $\nabla_{\mu} = \partial_{\mu} + \omega'_{\mu}$, where 
\begin{align} 
\omega_\mu' &= \frac{1}{2}(g_{\mu\nu}K_\nu + g_{\mu\nu}g^{\rho\sigma}\Gamma_{\rho\sigma}^{\nu})\label{eq:defomega}.
\end{align}
Using this $\omega'_\mu$ and $L$ we find $E \in \Gamma(\End(V))$ and compute for the curvature $\Omega_{\mu\nu}$ of $\nabla$:
\begin{subequations}
\begin{align}
  E &:= L - g^{\mu\nu}\partial_{\nu}(\omega_{\mu}') - g^{\mu\nu}\omega_\mu'\omega_\nu' + g^{\mu\nu}\omega_\rho'\Gamma_{\mu\nu}^\rho; \label{eq:defE} \\
 \Omega_{\mu\nu} &:= \partial_{\mu}(\omega'_{\nu}) - \partial_{\nu}(\omega'_{\mu}) - [\omega'_{\mu},\omega'_{\nu}]\label{eq:defOmega}.
\end{align}
\end{subequations}
In this situation we can make an asymptotic expansion (as $t \to 0$) of the trace of the operator $e^{-tP}$ in powers of $t$:
\begin{equation}
\tr\,e^{-tP} \sim \sum_{n \geq 0}t^{(n-m)/2}a_n(P),\qquad a_n(P) := \int_{M}a_n(x, P)\sqrt{g}d^m x\label{eq:gilkey},
\end{equation}
where $m$ is the dimension of $M$ and the coefficients $a_n(x, P)$ are called the \emph{Seeley--DeWitt coefficients}. It turns out \cite[Ch~4.8]{Gil84} that $a_n(x, P) = 0$ for $n$ odd and that the first three even coefficients are given by
\begin{subequations}
	\begin{align}
		a_0(x, P) &= (4\pi)^{-m/2}\tr(\id)\label{eq:gilkey1};\\
		a_2(x, P) &= (4\pi)^{-m/2}\tr(-R/6\,\id + E)\label{eq:gilkey2};\\
		a_4(x, P) &= (4\pi)^{-m/2}\frac{1}{360}\tr\big(-12R_{;\mu}^{\phantom{;\mu}\mu} + 5R^2 - 2R_{\mu\nu}R^{\mu\nu} \label{eq:gilkey3} \\
 &\qquad+ 2R_{\mu\nu\rho\sigma}R^{\mu\nu\rho\sigma} - 60RE + 180E^2 +60 E_{;\mu}^{\phantom{;\mu}\mu} + 30\Omega_{\mu\nu}\Omega^{\mu\nu}\big) \nonumber ,
	\end{align}
\end{subequations}
where $R_{;\mu}^{\phantom{;\mu}\mu} := \nabla^{\mu}\nabla_{\mu}R$ and the same for $E$. In all cases that we will consider, the manifold will be taken without boundary so that the terms $E_{;\mu}^{\phantom{;\mu}\mu}, R_{;\mu}^{\phantom{;\mu}\mu}$ vanish by Stokes' Theorem.

\medskip

This can be used in the computation of the spectral action as follows. Assume that the inner fluctuations give rise to an operator $D_A$ for which $D_A^2$ is of the form \eqref{eq:elliptic} on some vector bundle $V$ on a compact Riemannian manifold $M$. Then, on writing $f$ as a Laplace transform, we obtain
$$
f (D_A/\Lambda) = \int_{t>0} \tilde g(t) e^{-t D_A^2/\Lambda^2} ~ d t.
$$
In the case of a four-dimensional manifold the dominant terms of the expansion are found with Eq. \eqref{eq:gilkey} to be
  \begin{align}
    \tr f(D_A/\Lambda) &= 2\Lambda^4 f_4 a_0(D_A^2) + 2\Lambda^2 f_2 a_2(D_A^2) + a_4(D_A^2)f(0) + \mathcal{O}(\Lambda^{-2}),\label{eq:expansion_action_functional}
  \end{align}
where the $f_k$ are moments of the function $f$: 
\begin{align*}
	f_{k} := \int_{0}^{\infty} f(w)w^{k-1}dw; \qquad (k>0) \nonumber. 
\end{align*}

\section{The Einstein--Yang--Mills system}\label{ch:eym}

A spectral triple that will serve as the starting point for many of the subsequent considerations is the one that results in the Einstein--Yang--Mills system; it was introduced and studied in \cite{CC97} (cf. also \cite[Sect. 11.4]{CM07}). From now on, $M$ will denote a compact four-dimensional Riemannian spin manifold (with metric $g$). We take our spectral triple to be the tensor product of the canonical one on $(M,g)$, and the \emph{finite spectral triple} $(M_N(\mathbb{C}), M_N(\mathbb{C}), 0)$:
\begin{eqnarray}
  \A &=& C^{\infty}(M)\otimes M_N(\mathbb{C}) \simeq C^{\infty}(M, M_N(\mathbb{C}));\nonumber\\
 \H &=& L^2(M, S)\otimes M_N(\mathbb{C});\nonumber\\
     D &=& \slashed{\partial}_M\otimes \id\nonumber,\quad\text{with}\ \ \slashed{\partial}_M = i\gamma^{\mu}\nabla^S_{\mu},
\end{eqnarray}
where the representation of $M_N(\mathbb{C})$ on $M_N(\mathbb{C})$ is by left multiplication. We make the spectral triple \textit{real} by defining $J:\H\to\H$ by
\begin{equation}
 	J(s\otimes T) := J_M s\otimes T^*,\quad s \otimes T \in \H\label{eq:definition_J},
\end{equation}
where $J_M$ is the real structure on $L^2(M, S)$ ({i.e.} charge conjugation) and $T^*$ is the adjoint of the matrix $T$. 

The inner fluctuations \eqref{eq:definnerf} of this Dirac operator are seen to be of the form
\begin{equation}
 A + JAJ^* = \gamma^{\mu}\ad(A_\mu) \label{eq:fluctuations}, 
\end{equation}
where $\ad(A_\mu)T := [A_\mu, T], T \in M_N(\mathbb{C})$ and the minus sign giving rise to this commutator comes from the identity 
\begin{align}
J_M\gamma^{\mu}J_M^* &= J_M\gamma^{\mu}J_M^{-1} = -\gamma^{\mu}\label{eq:gammaJ_idn}.
\end{align}
The local expression for the fluctuated Dirac operator is then
\begin{equation}
  D_A = ie^{\mu}_{a}\gamma^{a}[(\partial_{\mu} + \omega_{\mu})\otimes\id + \id\otimes \mathbb{A}_\mu]\label{eq:dirac_full},
\end{equation}
where $\omega_\mu$ is the spin connection and $\mathbb{A}_\mu := -i \ad A_\mu$ is \emph{skew-Hermitian} due to the self-adjointness of $A_\mu$.\\

From the demand of self-adjointness of $D_A$, it follows that $\mathbb{A}$ is a $\mathfrak{u}(N)$-valued one-form.  
Now, $U(N)$ is not a simple group but $U(N) \simeq U(1)\times SU(N)$ resulting in $\mathfrak{u}(N) \simeq \mathfrak{u}(1)\oplus \mathfrak{su}(N)$ for the corresponding Lie algebras. But since $A + JAJ^{-1}$ is in the adjoint representation [cf.~\eqref{eq:fluctuations}] of $U(N)$, we retain only a traceless object. The symmetry group of the fluctuations is therefore effectively $SU(N)$.

\begin{prop}
\label{prop:squareD}
The square $D_A^2$ of the operator 
given in \eqref{eq:dirac_full} is of the form $ - \big(g^{\mu\nu}\partial_{\mu}\partial_{\nu} + K^{\mu}\partial_{\mu} + L)$ (cf. \eqref{eq:elliptic}) with
        \begin{align}
          K^{\mu} &= (2\omega^{\mu}-\Gamma^{\mu})\otimes\id + 2\,\id\otimes\mathbb{A}^{\mu}\nonumber\\
          L	  &= (\partial^{\mu}\omega_{\mu} + \omega^{\mu}\omega_{\mu} - \Gamma^{\mu}\omega_{\mu} + \tfrac{1}{4}R)\otimes\id + \id\otimes(\partial^{\mu}\mathbb{A}_{\mu} + \mathbb{A}^{\mu}\mathbb{A}_{\mu}) \nonumber\\
	           &\qquad + 2\,\omega^{\mu}\otimes \mathbb{A}_{\mu} - \Gamma^{\mu}\otimes\mathbb{A}_{\mu} - \tfrac{1}{2}\gamma^\mu\gamma^\nu\otimes \mathbb{F}_{\mu\nu}\nonumber,
        \end{align}
where $\Gamma^\nu = \Gamma^\nu_{\mu\lambda}g^{\mu\lambda}$ and $\mathbb{F}_{\mu\nu}$ is the curvature of the connection $\mathbb{A}_\mu$:
\begin{align}
   \mathbb{F}_{\mu\nu} := 
\partial_{\mu}\mathbb{A}_{\nu} - \partial_{\nu}\mathbb{A}_{\mu} + [\mathbb{A}_{\mu}, \mathbb{A}_{\nu}].  
\end{align}
\end{prop}
With this we can both determine $\omega'_{\mu}$ (and consequently $\Omega_{\mu\nu}$) and $E$ [cf.~\eqref{eq:defomega}, \eqref{eq:defOmega} and \eqref{eq:defE} respectively] uniquely:
\begin{align}
  \omega_{\mu}' = \omega_{\mu}\otimes\id + \id\otimes\mathbb{A}_{\mu} \label{eq:omegaaccent},\qquad
  E &= \tfrac{1}{4}R\otimes\id - \tfrac{1}{2}\gamma^\mu\gamma^\nu\otimes \mathbb{F}_{\mu\nu},\qquad
\Omega_{\mu\nu} = \tfrac{1}{4}R^{ab}_{\mu\nu}\gamma_{ab}\otimes\id + \id\otimes \mathbb{F}_{\mu\nu} \nonumber. 
\end{align}

We have shown that the fluctuated Dirac operator $D_A$ meets the demands needed to apply the heat kernel expansion of the spectral action, as sketched at the end of the previous section. Now for the first three coefficients appearing in \eqref{eq:expansion_action_functional} we have the following expressions: 
\begin{align}
   a_0(D_A^2) &= \frac{N^2}{4\pi^2}\int_M\volelt,\\
   a_2(D_A^2) &= \frac{N^2}{48\pi^2}\int_{M}R\volelt\\ 
   a_4(D_A^2) &= \frac{1}{16\pi^2}\frac{N^2}{360}\int_{M}\bigg[5R^2 - 8R_{\mu\nu}R^{\mu\nu} - 7R_{\mu\nu\rho\sigma}R^{\mu\nu\rho\sigma}\bigg]- \frac{1}{24\pi^2}\int_{M} \tr(\mathbb{F}_{\mu\nu}\mathbb{F}^{\mu\nu})\volelt\label{a4final}
\end{align}
where $N^2$ originates from $\tr_{M_N(\mathbb{C})}\id$. Inserting these expressions into \eqref{eq:expansion_action_functional} then results in
\begin{thm}[Chamseddine--Connes \cite{CC97}]
\label{thm:bos-action-EYM}
The bosonic action for the inner fluctuations of the spectral triple $(C^\infty(M, M_N(\C)), L^2(M, S)\otimes M_N(\C), \dirac \otimes 1)$ is given by
\begin{align*}
S_b[A]\equiv  \tr f(D_A/\Lambda) &= \frac{1}{4\pi^2}\int_{M}
\mathcal{L}_b(g, A)\volelt  + \mathcal{O}(\Lambda^{-2})
,
\end{align*}
with Lagrangian 
\begin{align*}
	\mathcal{L}_b(g, A) &= 2f_4\Lambda^4N^2 + \frac{N^2}{6}f_2\Lambda^2R + \frac{f(0)N^2}{1440}\bigg[5R^2 - 8R_{\mu\nu}R^{\mu\nu}-7R_{\mu\nu\rho\sigma}R^{\mu\nu\rho\sigma}\bigg]- \frac{f(0)}{6}\tr(\mathbb{F}_{\mu\nu}\mathbb{F}^{\mu\nu}).\nonumber
\end{align*}
\end{thm}
This expression contains both the Einstein--Hilbert action of General Relativity and the Yang--Mills action of a $SU(N)$-gauge field. Since the term $\inpr{\psi}{D_A\psi}$ accounts for the fermionic propagator and interactions of the fermion $\psi$ with the gauge field, the sum
\begin{equation*}
S[A,\psi] := S_b[A]+ S_f[A,\psi] =  \tr f(D_A/\Lambda)+\inpr{\psi}{D_A\psi} 
\end{equation*}
gives the full action of the Einstein--Yang--Mills system plus terms of order $\Lambda^{-2}$. The gauge group $SU(\A) = C^\infty(M, SU(N))$ acts on the gauge potential $A$ and on $\psi$ in the adjoint representation.

\section{Supersymmetry in the Einstein--Yang--Mills system}\label{ch:susy_in_EYM}

We would like to obtain a realization of supersymmetry for the Einstein--Yang--Mills system, as considered in the previous section, in the framework of noncommutative geometry. The possibility of such a supersymmetry was suggested in \cite{CC97}. 

We work this out in full detail and give the supersymmetry transformation establishing this symmetry between the fermionic and bosonic fields.
\\

In trying to do so, we immediately stumble upon the problem that bosonic and fermionic fields do not have the same number of degrees of freedom, as is required for supersymmetry. Indeed, both in the spinorial as in the finite part the fermionic degrees of freedom exceed those of the bosons: by requiring self-adjointness and unimodularity, the finite part of the bosons was seen to be $\mathfrak{su}(N)$-valued one-forms. The finite part of the fermions, on the other hand, is an element of $M_N(\mathbb{C})$. On top of that, a spinor $\psi(x)$ has eight real (four complex) degrees of freedom whereas the continuous part of the gauge potential has only four: $A_\mu$, $\mu = 1, \ldots, 4$.\\

We will solve these two problems one by one in the subsequent subsections.

\subsection{Majorana and Weyl fermions}

The basic fermionic constituents of most supersymmetric theories are Majorana fermions; particles that are invariant under charge conjugation. However, in this Euclidean set up, we have $J^2 = -1$ with which only $\psi(x) = 0$ could be Majorana. 
Indeed (massless) Majorana fermions do not exist in a 4 dimensional Euclidean space, as was pointed out by Schwinger  \cite{Sch59} already in 1959.\\
An alternative way to correctly reduce the number of degrees of freedom is to restrict the input of the inner product to eigenspaces $\H^\pm$ of $\gamma$. To this end Chamseddine, Connes and Marcolli \cite{CCM07} propose  as a fermionic action
$\frac{1}{2}\inpr{J\psi}{D_A\psi}$ instead of $\inpr{\psi}{D_A\psi}$. This would be of no avail to us, since this allows for such a restriction only when $J\gamma = \gamma J$, in accordance with the classification of \cite{CC07a,CC07b}. In our case, it would automatically yield $\inpr{J\psi_1}{D\psi_2} = - \inpr{J\psi_1}{D\psi_2}$ for all $\psi_1, \psi_2 \in \H$.\\

Different but similar solutions of this problem were given by Van Nieuwenhuizen and Waldron \cite{NW96} and Nicolai \cite{Nic78}. To obtain a Lagrangian in Euclidean space, whose Green functions are analytic continuations of their Minkowskian counterpart, Van Nieuwenhuizen and Waldron propose the following. Starting from a Lagrangian for a single Weyl fermion, they define a Wick rotation on the spinors themselves. When applying this, one is obliged to drop the Minkowskian reality constraint $\overline{\psi} := \psi^\dagger \gamma^0$ ---rotating $\overline{\psi}$ and $\psi$ separately. The result is then a Lagrangian containing Weyl spinors $\chi$ and $\psi$ of opposite chirality. Since the system still contains two fermionic variables ($\chi$ and $\psi$ instead of $\overline{\psi}$ and $\psi$) the path integral is insensitive to such a rotation. The solution is thus to take as the fermionic part of the action
\begin{align}
  S_f[ A,\psi, \chi] := \inpr{\chi}{D_A\psi}; \qquad (\psi \in \H^+, \chi \in \H^-)\label{eq:fermion_inner_product2},
\end{align}
which is the Euclidean counterpart of the action for $\overline{\psi}$ and $\psi$ in Minkowskian space. 

\subsection{Unimodularity for fermions}

The reduction from $M_N(\mathbb{C})$ to $\mathfrak{su}(N)$ takes place in two steps; first from $M_N(\mathbb{C})$ to $\mathfrak{u}(N)$ and second from $\mathfrak{u}(N)$ to $\mathfrak{su}(N)$.\\

For the first part we simply use the fact that the $M_N(\mathbb{C})$ is the complexification of $\mathfrak{u}(N)$:
$
 M_N(\mathbb{C}) \simeq \mathbb{C}\otimes_{\mathbb{R}}\mathfrak{u}(N)\nonumber
$. For the full Hilbert space $\H$ this implies already that
\begin{align}
 \H = L^2(M, S)\otimes_{\mathbb{C}} M_N(\mathbb{C}) \simeq L^2(M, S)\otimes_{\mathbb{R}}\mathfrak{u}(N)\nonumber.
\end{align}

We obtain the reduction from $\mathfrak{u}(N)$ to $\mathfrak{su}(N)$ by splitting any fermion into a \emph{trace} and a \emph{traceless} part:
	 \begin{align*}
	   \widetilde{\psi} = \tr \widetilde{\psi} + \psi \in L^2(M, S)\otimes \big(\mathfrak{u}(1) \oplus \mathfrak{su}(N)\big).
	 \end{align*}
Inserting this expression into the inner product, we get
	 \begin{align*}
	   \inpr{\widetilde{\chi}}{D_A\widetilde{\psi}} &= \inpr{\tr\widetilde{\chi}}{D_A \tr\widetilde{\psi}} + \inpr{\chi}{D_A \tr\widetilde{\psi}}\nonumber
 + \inpr{\tr\widetilde{\chi}}{D_A\psi} + \inpr{\chi}{D_A\psi}
	   = \inpr{\tr \widetilde{\chi}}{D\tr\widetilde{\psi}} + \inpr{\chi}{D_A\psi},
	 \end{align*}
where we have used that for $\lambda \in \mathfrak{u}(1)$ and $X, X_1, X_2 \in \mathfrak{su}(N)$ that $[X, \lambda] = 0 = \tr(X \lambda )$. 
So the two different parts decouple and the trace-part lacks any gauge interactions; it describes a totally free fermion. We therefore discard it from the theory. 

\subsection{Supersymmetry transformations}\label{ch:seym}

After the preparations done in the two previous subsections, the Einstein--Yang--Mills system is at least suited for supersymmetry. What is left, is actually proving that the system is supersymmetric.

Thus, consider the action $S[A,\psi,\chi] = S_b[A] + S_f[A,\psi,\chi]$ in terms of the two traceless Weyl spinors $\psi$ and $\chi$ and the $SU(N)$-gauge field $\bA$. We conveniently write the fermionic action,
$$
S_f[A, \psi, \chi] = \langle \chi, D_A \psi \rangle = \int_M \tr_F ( \chi, D_A \psi ) \sqrt{g} d^4x,
$$
in terms of a Hermitian pairing \mbox{$(.\,,.):\Gamma^{\infty}(S) \times \Gamma^{\infty}(S) \to C^{\infty}(M)$} and a trace $\tr_F$ over the \emph{finite part}.\\

In order to see whether this system exhibits supersymmetry, we will define
\begin{align*}
  \delta \mathbb{A} \in B(\H)\quad \text{and}\quad
  \delta \psi \in \H^+, \delta \chi \in \H^-
\end{align*}
---where the expressions for $\delta\mathbb{A}, \delta\psi$ and $\delta\chi$ contain their respective \emph{superpart\-ners}--- under which the action is invariant:
\begin{align}
  \delta S[A,\psi, \chi] := \frac{d}{dt}S[ A + t\delta A, \psi + t\delta\psi, \chi + t\delta\chi]\bigg|_{t=0} = 0\label{eq:transf_susy}.
\end{align}

From here on $\epsilon_{\pm}$ will denote a pair of $\gamma^5$ eigenspinors that are singlets of the gauge group and vanish covariantly: $\nabla^S_\mu \epsilon_{\pm} = 0$.
\begin{defin}\label{def:susytransf} For $A \in \Omega^1_D(\A)$, $\psi, \in \H^+, \chi \in \H^-$ we define $\delta A \in \B(\H),\delta\psi \in \H^+, \delta\chi \in \H^-$ by
\begin{align*}
  \delta A &:= \gamma^\mu [c_1(\epsilon_-, \gamma_\mu\psi) + c_2(\chi, \gamma_\mu\epsilon_+)]
,\\
  \delta\psi &:= c_3F\epsilon_+ \qquad \text{and}\qquad
  \delta\chi := c_4F\epsilon_-
\end{align*}
where $F \equiv \gamma^\mu\gamma^\nu F_{\mu\nu}^a \otimes T_a$, $F_{\mu\nu} = \partial_\mu A_\nu - \partial_\nu A_\mu + [A_\mu,A_\nu]$ and $c_{1}\ldots c_4 \in \mathbb{R}$.
\end{defin}
The constants $c_{1}\ldots c_4$ are yet to be determined.

\begin{prop}\label{prop:susy_fermions}
With the definitions given above, we have for the fermionic part of the action
  \begin{align*}
    \delta S_f[A, \psi, \chi] = 
    -2 ic_4 \inpr{\epsilon_-}{F_{\mu\nu}\gamma^\mu D^{\nu}\psi} - 2ic_3\inpr{F_{\mu\nu}\gamma^\mu D^{\nu}\chi}{ \epsilon_+}
.
  \end{align*}
where $D_\mu = \nabla_\mu^S + A_\mu$ is the covariant derivative.
\end{prop}
\begin{proof}
  We apply the above supersymmetry transformations to the fermionic part of the action to obtain
\begin{align}
\delta S_f[A,\psi, \chi]
  &= \frac{d}{dt}\inpr{\chi + t\delta\chi}{D_{A + t\delta A}(\psi + t\delta\psi)}\Big|_{t = 0}
  = c_4\inpr{F\epsilon_-}{D_A\psi} + \inpr{\chi}{\delta\mathbb{A}\psi} + c_3\inpr{\chi}{D_AF\epsilon_+}  \label{eq:susy_fermion0}.
\end{align}
Let us look at the terms on the right hand side one by one. Writing out $F$, using the self-adjointness of $D_A$ and the identity 
\begin{align*}
  \nabla^S_\nu \gamma^\mu = - \Gamma^\mu_{\nu\alpha}\gamma^\alpha
\end{align*}
twice, we get for the first term
\begin{align}
  c_4\inpr{F\epsilon_-}{D_A\psi} &= - c_4i\inpr{\Gamma^\mu_{\sigma\alpha}\gamma^\sigma[\gamma^\alpha, \gamma^\nu]F_{\mu\nu}\epsilon_-}{\psi} + c_4i\inpr{\gamma^\sigma\gamma^\mu\gamma^\nu D_\sigma F_{\mu\nu}\epsilon_-}{\psi}\label{eq:susy_fermion0.5}.
\end{align}
Using the identity 
\begin{align}
  \gamma^\mu\gamma^\nu\gamma^\sigma = g^{\mu\nu}\gamma^\sigma - g^{\sigma\nu}\gamma^\mu + g^{\mu\sigma}\gamma^\nu - i\epsilon^{\mu\nu\sigma\lambda}\gamma^5\gamma_\lambda\nonumber
\end{align}
and the symmetry $\Gamma^\mu_{\nu\lambda} = \Gamma^\mu_{\lambda\nu}$ of the Christoffel symbols, the first term on the RHS of \eqref{eq:susy_fermion0.5} is seen to vanish whereas the second term now reads
\begin{align}
  (F\epsilon_-, D_A\psi) &= i([g^{\sigma\mu}\gamma^\nu - g^{\sigma\nu}\gamma^\mu +g^{\mu\nu}\gamma^\sigma -i\epsilon^{\sigma\mu\nu\lambda}\gamma^5\gamma_\lambda]D_\sigma F_{\mu\nu}\epsilon_-, \psi) \label{eq:susy_fermion1}. 
\end{align}
The first two terms of \eqref{eq:susy_fermion1} add up by the antisymmetry of $F_{\mu\nu}$, whereas the third term vanishes for that very reason. The fourth term vanishes in view of the Bianchi identity:
\begin{align}
  [D_\mu F_{\nu\sigma} + D_\sigma F_{\mu\nu} + D_\nu F_{\sigma\mu}](x) = 0\quad\forall\ x \in M\nonumber. 
\end{align}
We are thus left with:
\begin{align}
  c_4\inpr{F\epsilon_-}{D_A\psi} = 2c_4i(\gamma^\nu D^\mu F_{\mu\nu}\epsilon_-, \psi)\nonumber.
\end{align}
By exactly the same reasoning we can rewrite the third term of \eqref{eq:susy_fermion0}. Now we are still left with the second term of \eqref{eq:susy_fermion0}, which yields for each point $x \in M$:
\begin{align}
  \tr_F(\chi, \delta\mathbb{A}\psi)(x) &= f_{abc}(\chi^a, \gamma^\mu\psi^c)[c_1(\epsilon_-, \gamma_\mu\psi) + c_2(\chi, \gamma_\mu\epsilon_+)](x).
\end{align}
Both terms are seen to vanish separately using the antisymmetry of $f_{abc}$ and a Fierz transformation (see Appendix \ref{se:fiers_trans} for details). Adding the results for the first and third terms of \eqref{eq:susy_fermion0} yields the expression:
\begin{align*}
  \delta S_f[A,\psi, \chi] &= 2c_4i\inpr{\gamma^\nu D^\mu F_{\mu\nu}\epsilon_-}{\psi} + 2c_3i\inpr{\chi}{ \gamma^\nu D^\mu F_{\mu\nu}\epsilon_+}
\qedhere
\end{align*}
\end{proof}

That covered the fermionic part of the action. Regarding the bosonic part we can see that after performing the supersymmetry transformations 
\begin{prop}
  The square of the operator $D_{A + t\delta A}$ with $\delta A$ given in Definition \ref{def:susytransf} is of the form in \eqref{eq:elliptic}:
  \begin{align*}
    D_{A + t\delta A}^2 &= -[g_{\mu\nu}\partial^\mu\partial^\nu + {K'}^\mu\partial_\mu + L']
  \end{align*}
with ${K'}^\mu$ and $L'$ given in terms of the $K^\mu$ and $L$ of Proposition \ref{prop:squareD} as
\begin{align*}
{K'}^\mu &= K^\mu - 2c_1\otimes\ad(\epsilon_-, \gamma^\mu\psi)t,\nonumber\\
L' &= L - c_1\gamma^\nu\gamma^\mu\otimes\ad(\epsilon_-, \gamma_\mu D_\nu\psi)t - 2c_1\otimes\ad(\epsilon_-, \gamma^\nu\psi)[(\omega_\nu - \Gamma_\nu)\otimes 1 + \mathbb{A}_\nu t] + \mathcal{O}(t^2),
\end{align*}
\end{prop}
\begin{proof}
We will explicitly calculate 
\begin{align}
  D_{A + t\delta A}^2 &= [D + A + t\delta A + J(A + t\delta A)J^*]^2= D_A^2+ i\{D_A, \delta \mathbb{A}\}t + \cO(t^2)\label{eq:DiracSquared},
\end{align}
where $\delta A + J\delta AJ = i\delta\mathbb{A}$. Let us for the moment ignore the term in $\delta A$ depending on $\chi$.

Then, the second term on the right hand side of \eqref{eq:DiracSquared} reads 
\begin{align}
  \{D_A, \delta \mathbb{A}\} &= [i\gamma^\nu(\nabla^S_\nu \otimes \id + \mathbb{A}_\nu)][c_1\gamma^\mu \otimes \ad(\epsilon_-, \gamma_\mu\psi)]+ [ic_1\gamma^\mu\otimes\ad(\epsilon_-, \gamma_\mu\psi)][\gamma^\nu (\nabla^S_\nu\otimes \id + \mathbb{A}_\nu)]\nonumber\\
    &= ic_1\gamma^\mu\gamma^\nu(\partial_\mu\otimes 1 + \mathbb{A}_\mu)\ad(\epsilon_-, \gamma_\nu\psi) - ic_1\gamma^\mu\gamma^\nu\ad(\epsilon_-, \gamma_\nu\psi)(\partial_\mu\otimes 1 + \mathbb{A}_\mu)\nonumber\\
    &\hspace{3cm} + 2ic_1\ad(\epsilon_-, \gamma^\mu\psi)(\nabla^S_\mu + \mathbb{A}_\mu) - ic_1\Gamma^\nu_{\mu\lambda}\gamma^\mu\gamma^\lambda\ad(\epsilon_-, \gamma_\nu\psi)\label{eq:DiracSquared2}, 
\end{align}
where the Christoffel symbol stems from interchanging the spin connection $\nabla^S_\mu$ with a gamma matrix.
Using that $\epsilon_-$ vanishes covariantly we have
\begin{align*}
  \partial_\mu(\epsilon_-, \gamma_\nu\psi) = (\epsilon_-, \gamma_\nu\nabla^S_\mu\psi) - \Gamma^{\lambda}_{\mu\nu}(\epsilon_-, \gamma_\lambda\psi) + (\epsilon_-, \gamma_\nu\psi)\partial_\mu.
\end{align*}
Inserting this expression into \eqref{eq:DiracSquared2} and using the definition $\Gamma^\mu = g^{\nu\lambda}\Gamma^\mu_{\nu\lambda}$, we receive for $\{D_A, \delta\mathbb{A}\}$:
\begin{align*}
  \{D_A, \delta\mathbb{A}\} &= ic_1\gamma^\mu\gamma^\nu\ad(\epsilon_-, \gamma_\nu D_\mu\psi) + 2ic_1\ad(\epsilon_-, \gamma^\nu\psi)[(\partial_\mu + \omega_\mu - \Gamma_\mu)\otimes 1 + \mathbb{A}_\mu].
\end{align*}
Plugging this into \eqref{eq:DiracSquared} yields the desired form of ${K'}^\mu$ and $L'$. 
\end{proof}

We are thus allowed to perform a heat kernel expansion \eqref{eq:gilkey} for $D_{A + t\delta A}$ to see to what extent each of the coefficients $a_n(D_A^2)$ (for $n = 0, 2, 4$) is invariant under supersymmetry. \\

The objects that are of interest to us are  
\begin{align}
\delta a_n (D_A^2) := \frac{\mathrm{d}}{\mathrm
{d}t} a_n(D_{A+t\delta A}^2)\Big|_{t=0} &= \frac{d}{dt} a_n(D_{A}^2+\{\delta \mathbb{A}, D_A\}t + \mathcal{O}(t^2))\Big|_{t=0}\label{eq:susy_a_n}, 
\end{align}
the first of which are given by \eqref{eq:gilkey1}, \eqref{eq:gilkey2} and \eqref{eq:gilkey3}. The $E$ and $\Omega_{\mu\nu}$ appearing in these formulas are of course different than before, but still related to $K^\mu$ and $L$ (as given above) in the same way; by \eqref{eq:defE} and \eqref{eq:defomega}. 
Short calculations show that the changes of $K^\mu$ to ${K'}^\mu$ and $L$ to $L'$ have the following effect on the variable $E$ and $\Omega_{\mu\nu}$:
\begin{align}
E' &= E + c_1\gamma^\mu\gamma^\nu\otimes\ad(\epsilon_-,\gamma_\nu D_\mu\psi)t - c_1\id\otimes\ad(\epsilon_-, \gamma^\mu D_\mu\psi)t + \mathcal{O}(t^2)\nonumber\\
\Omega_{\mu\nu}' &= \Omega_{\mu\nu} + c_1\id\otimes[\ad(\epsilon_-, \gamma_{\nu}D_{\mu}\psi) - c_1\ad(\epsilon_-, \gamma_{\mu}D_{\nu}\psi)]t + \mathcal{O}(t^2)\nonumber. 
\end{align}
Having found these particular expressions, we are ready to determine \eqref{eq:susy_a_n}.
\begin{prop}
The Seeley--DeWitt coefficients $a_0(D_A^2)$ and $a_2(D_A^2)$ are invariant under the supersymmetry transformation $A \mapsto A+t \delta A$ given in Definition \ref{def:susytransf}, whereas $a_4(D_A)$ transforms as
  \begin{align*}
    \frac{d}{dt} a_4(D_{A+t\delta A}^2)\Big|_{t=0} &= \frac{c_1}{6\pi^2}\inpr{\epsilon_-}{F^{\mu\nu}\gamma_\nu D_\mu\psi} + \frac{c_2}{6\pi^2}\inpr{F^{\mu\nu} \gamma_\nu D_\mu \chi}{\epsilon_+} ,  
  \end{align*}
\end{prop}
\begin{proof}
  The first coefficient $a_0(D_{A+t\delta A}^2)$ is trivial: the identity does not transform under supersymmetry. Ignoring for the moment the term linear in $\chi$ for the second Seeley--DeWitt coefficient, there is only one contribution [see \eqref{eq:gilkey2}];
\begin{align}
 \frac{d}{dt} \tr (E')\bigg|_{t=0} &= c_1i\tr[\gamma^\mu\gamma^\nu\otimes\ad(\epsilon_-,\gamma_\nu D_\mu\psi) - \id\otimes\ad(\epsilon_-, \gamma^\mu D_\mu \psi)]= 0.
\end{align}
For the third coefficient \eqref{eq:gilkey3} there are three terms of interest 
\begin{align}
 \delta a_4(D_A^2) &= \frac{1}{192\pi^2}\frac{d}{dt}\int_M\tr\big(6{E'}^2 + {\Omega'}_{\mu\nu}{\Omega'}^{\mu\nu} + 2RE'\big)\sqrt{g}d^4x\bigg|_{t=0}\label{eq:susya4},
\end{align}
where the last one vanishes by the same reasoning as employed above. For the first term we use that $E = \tfrac{1}{4}R\otimes 1 -\tfrac{1}{2}\gamma^\mu\gamma^\nu\otimes\mathbb{F}_{\mu\nu}$ and obtain
 \begin{align}
    \frac{d}{dt}\int_M\tr{E'}^2\sqrt{g}d^4x\bigg|_{t=0} &= -2\frac{1}{2}c_1\int_M\bigg[\tr\big[\gamma^\lambda\gamma^\sigma\gamma^\mu\gamma^\nu\otimes\mathbb{F}_{\lambda\sigma}\ad(\epsilon_-,\gamma_\nu D_\mu\psi)\nonumber\\
    &\qquad - \gamma^\mu\gamma^\nu\otimes\mathbb{F}_{\mu\nu}\ad(\epsilon_-, \gamma^\lambda D_\lambda\psi)\big]\bigg]\sqrt{g}d^4x\nonumber\\
  &= -4c_1\big(\delta^{\lambda\nu}\delta^{\sigma\mu} - \delta^{\lambda\mu}\delta^{\nu\sigma}\big)
\times\int_M\tr[\mathbb{F}_{\lambda\sigma}\ad(\epsilon_-,\gamma_\nu D_\mu\psi)]\sqrt{g}d^4x\nonumber\\
   &= 8c_1\int_M\tr\big[\mathbb{F}^{\mu\nu}\ad(\epsilon_-,\gamma_\nu D_\mu\psi)]\sqrt{g}d^4x
= -8Nc_1\inpr{\epsilon_-}{F^{\mu\nu}\gamma_\mu D_\nu\psi}\nonumber,
 \end{align}
where at various points we have used that $\mathbb{F}$ is antisymmetric. For the second term in \eqref{eq:susya4} we have with $\Omega_{\mu\nu} = 1\otimes\mathbb{F}_{\mu\nu} + \tfrac{1}{4}R^{ab}_{\mu\nu}\gamma_{ab}\otimes 1$
\begin{align}
\frac{d}{dt}\int_M\tr{\Omega'}_{\mu\nu}^2\sqrt{g}d^4x\bigg|_{t=0} &= 2c_1\int_M\tr\big[ 1\otimes\mathbb{F}_{\mu\nu}+ \tfrac{1}{4}R^{ab}_{\mu\nu}\gamma_{ab}\otimes 1][\id\otimes\ad(\epsilon_-, \gamma_{[\nu}D_{\mu]}\psi)]\big]\sqrt{g}d^4x\nonumber\\
&= 16c_1\int_M\tr\big[\mathbb{F}_{\mu\nu}\ad(\epsilon_-, \gamma^{\nu}D^{\mu}\psi)\big]\sqrt{g}d^4x= -16Nc_1\inpr{\epsilon_-}{F_{\mu\nu}\gamma^\mu D^\nu\psi}\nonumber.
\end{align}
We thus get for \eqref{eq:susya4}:
  \begin{align*}
    \delta a_4(D_A^2) &= -\frac{c_1}{192\pi^2}(48 + 16)N\inpr{\epsilon_-}{F^{\mu\nu}\gamma_\mu D_\nu\psi}
= -\frac{c_1}{3\pi^2}\inpr{\epsilon_-}{F^{\mu\nu}\gamma_\mu D_\nu\psi}
  \end{align*}
and a similar term involving $\chi$, as required.
\end{proof}

Thus, although $a_0(D^2_A)$ (proportional to $\Lambda^4$) and $a_2(D_A^2)$ (proportional to $\Lambda^2$) are supersymmetry invariants, $a_4(D_A^2)$ transforms to an expression that equals the one of the fermionic action (cf. Proposition \ref{prop:susy_fermions}) by the right choice of coefficients. This means that
\begin{thm} 
The action $S[A, \psi,\chi]= S_b[A] + S_f[A,\psi,\chi]$ (with $S_f$ defined in \eqref{eq:fermion_inner_product2}) of the Einstein--Yang--Mills system is invariant under supersymmetry for at least all positive powers of $\Lambda$, provided
$$
c_4 = - \frac{i c_1 f(0)}{6 \pi^2} \qquad \text{and} \qquad  c_3 = - \frac{i c_2 f(0)}{6 \pi^2}.
$$
\end{thm} 

Though these results are encouraging, it is still somewhat unsatisfactory that we had to resort to a heat kernel expansion; a question whether or not the \emph{full} spectral action is supersymmetry invariant remains to be answered. As was noted by Chamseddine in \cite{Cha98}, noncommutative geometry treats bosons (spectral action) and fermions (inner product) on a different footing. Hence, any attempt (such as \cite{Sit08}) that combines both the inner product and the spectral action into a single expression is well worth studying from the perspective of supersymmetry.

\section{Supersymmetric QCD}\label{ch:superQCD} 

In this section, we consider a supersymmetric version of QCD ---the theory of quarks and gluons. For that we will be regarding only one of three generations of particles and we leave all leptons and electroweak gauge bosons out. 

\subsection{The finite spectral triple}
  If we want any chance of finding supersymmetry, we need to enlarge the finite part of the Hilbert space such that it contains not only the quarks and antiquarks, but the gluinos\footnote{We will postpone the (partial) justification of this terminology until \S 14.4.} ---the supersymmetric partners of the gluons and therefore fermions--- as well. Moreover, in order to keep the gauge group to be $SU(3)$ the algebra in our spectral triple should be $M_3(\C)$.

  \begin{defin}\label{def:sQCD_spectraltriple}
The finite spectral triple $(\A_F, \H_F, D_F)$ is defined by
\begin{itemize}
       \item[-] $\A_F := M_3(\mathbb{C})$.
       \item[-] $\H_F := \C^3 \oplus M_3(\mathbb{C})\oplus \C^{3}$, carrying the following representation of $\A_F$:
  \begin{align*}
    \pi(m)(q, g, q') = (mq, mg, q')\label{eq:squark_representation}.
  \end{align*}
       \item[-] $D_F$ is defined as the matrix:
          \begin{align*}
            D_F := \begin{pmatrix} 0 & d & 0 \\ d^* & 0 & e^* \\ 0 & e & 0 \end{pmatrix},
          \end{align*}
          with \mbox{$d : M_3(\mathbb{C}) \to \C^3$} and \mbox{$e : M_3(\mathbb{C}) \to \C^{3}$} arbitrary linear maps. 
     \end{itemize}
  \end{defin}

The conditions of a spectral triple are trivially satisfied; we would like to define a real structure $J_F$ on it as well. Our candidate is
\begin{align}
              J_F(q_1, g, q_2) := (\overline{q_2}, g^*, \overline{q_1}) \in \C^3\oplus M_3(\mathbb{C})\oplus\C^{3},
           \end{align}

This form of $J_F$, as with the representation of the algebra, is as expected: on the two copies of $\C^3$  it is ---up to interchanging these two copies--- the same as in the noncommutative description of the Standard Model \cite{CCM07}; on $M_3(\mathbb{C})$ it is the same as in the Einstein--Yang--Mills system. 

  \begin{lem}\label{lem:d_determines_e}
    With $J_F$ as above, the requirement $D_FJ_F = J_FD_F$ uniquely determines $e$ in terms of $d$:
    \begin{align*}
      \overline{e(g)} &= d(g^*)\qquad\forall\ g \in M_3(\mathbb{C})
    \end{align*}
  \end{lem}
  \begin{proof}
This follows from a direct computation of $D_FJ_F$ and $J_FD_F$ acting on an element in $\H_F$.
  \end{proof}

Let us check the other conditions for a real spectral triple (of KO-dimension $0$). We compute for the \emph{opposite representation} $\pi^\circ(m) = J_F\pi(m^*)J_F^*$:
  \begin{align}   
      \pi^\circ(m)(q_1, g, q_2) :=  (q_1, gm, m^tq_2)\label{eq:squark_opposite_representation}.
  \end{align}
One easily checks that $\pi(m)$ commutes with $\pi^\circ(m')$ for any $m,m' \in M_3(\C)$, thus fulfilling the second condition in Equation \eqref{eq:order_one_condition}. In order to satisfy the first (i.e. the first-order condition), we make the following choice of $D_F$ in terms of a map $d : M_3(\mathbb{C}) \to \C^{3}$ of the form
  \begin{align}
    d(g) = gv\qquad\forall\ g \in M_3(\mathbb{C})\label{eq:def_mapd}
  \end{align}
for a fixed $v \in \C^3$.
This definition for $d$ corresponds to $d^*(q) = q\overline{v}^t$ (considered as the $3 \times 3$ matrix $(q\overline{v}^t)_{ij} = q_i\overline{v}_j$) for the adjoint of $d$ and by Lemma \ref{lem:d_determines_e} to
  \begin{align*}
    e(g) = g^t\overline{v}, \qquad
    e^*(q) = vq^t ; \qquad (g \in M_3(\mathbb{C}), q \in \C^3).
  \end{align*}
  for the map $e$ and its adjoint. 
 
  \begin{prop}
Given the representations of the algebra \eqref{eq:squark_representation} and \eqref{eq:squark_opposite_representation}, the finite Dirac operator with $d$ and $e$ as above, satisfies the order one condition \eqref{eq:order_one_condition}. Consequently, the set of data $(\A_F, D_F, \H_F, J_F)$ defines a finite real spectral triple of KO-dimension 0.
  \end{prop}
  \begin{proof}
    Writing out \eqref{eq:order_one_condition} and using that $\pi(m) = 1$ on antiquarks and $\pi^\circ(m) = 1$ on quarks gives a number of simultaneous demands:
    \begin{align*}
        d(mg) &= md(g), & d^*(mq) &= md^*q \nonumber\\
        e(gm) &= m^te(g) & e^*(m^tq') &= e^*(q')m\qquad\forall m \in m, g \in M_3(\mathbb{C}), q,q' \in \C^3. 
    \end{align*}
    These are easily seen to be met for the given representations and maps $d$ and $e$.
  \end{proof}

As a preparation for the next section, we determine the inner fluctuations of the finite Dirac operator, as well as the (finite) gauge group and its action.
  \begin{lem} 
\label{lma:inner-finite} 
The inner fluctuations $D_F + A_F + J_F A_F^* J_F^*$ with $A_F \in \Omega_{D_F}(\A_F)$ of the finite Dirac operator $D_F$ are parametrized by a vector $\tilde q \in \C^3$ as
  \begin{align*}
D_{\squark} :=   D_F + A_F + J_FA_F^*J_F^* &= g_3\begin{pmatrix} 0 & A_{\squark} & 0 \\ A^*_{\squark} & 0 & B^*_{\antisquark} \\
              0 & B_{\antisquark} & 0
           \end{pmatrix}
  \end{align*}
  with $g_3$ the QCD-coupling constant and
  \begin{align}
      A^*_{\squark}(q) &=q\antisquark^{\,t}, & A_{\squark}(g) &= g\squark,\nonumber\\  
      B^*_{\antisquark}(q) &= \squark q^t, & B_{\antisquark}(g) &= g^t\antisquark\nonumber.
  \end{align}
  \end{lem}
  \begin{proof}
  We have $A_F := \sum_i \pi(m_i)[D_F, \pi(n_i)]$ [cf.~\eqref{eq:innerfluctuationform}] which, applied to an element $(q_1, g, q_2) \in \H_F$, gives
  \begin{align}
     \sum_i \pi(m_i)[D_F, \pi(n_i)](q_1, g, q_2)
    &= \sum_i\big(0, m_i[1 - n_i]vq_2^t, g^t\overline{(n_i^* - 1)v})\big)\label{eq:squarkfluctuations}
  \end{align}
  For the other part, $J_FA_F^*J_F^*$, we compute
   \begin{align}  
      \sum_i J_F(\pi(m_i)[D_F, \pi(n_j)])^*J_F^* &= - \sum_i J_F[D_F, \pi(n_i^*)]\pi(m_i^*) J_F = -\sum_i [D_F, \pi^\circ(n_i)]\pi^\circ(m_i)\nonumber, 
  \end{align} 
  where we have used that $D_FJ_F = J_FD_F$. We therefore get
  \begin{align}
    J_FA_F^*J_F(q_1, g, q_2) 
    &= \sum_i \big(gm_i(1 - n_i)v, q_1\overline{(n_i^* - 1)v}^t, 0\big)\label{eq:squarkfluctuations_J}.
  \end{align}
    Requiring $A_F$ to be self-adjoint yields the demand
  \begin{align}
    \sum_i n_i^* - 1 &= \sum_i m_i(1 - n_i), 
  \end{align}
  for the elements of the algebra. Defining 
  \begin{align}
    \squark &:= g_3^{-1}\sum_i [1 + m_i(1 - n_i)]v = g_3^{-1}\sum_i n_i^*v\label{eq:def_squark}.
  \end{align}
  and adding the expressions \eqref{eq:squarkfluctuations} and \eqref{eq:squarkfluctuations_J} for $A_F$ and $J_FA_F^*J_F^*$ respectively to that of $D_F$, we get 
   \begin{align*}
    (D_F + A_F + J_FA_F^*J_F^*)(q_1, g, q_2) &= g_3\big(g\squark, q_1\antisquark^{\,t} + \squark q_2^{t}, g^t\antisquark\big)
  \end{align*}  
which is of the desired form.
  \end{proof}

\subsection{The product geometry and its inner fluctuations}
We next consider the product of the canonical spectral triple $(C^\infty(M), L^2(M,S),\dirac_M)$ associated to a four-dimensional Riemannian spin manifold $M$, with the above spectral triple $(\A_F,\H_F,D_F)$. Explicitly, we have
\begin{align*}
  \A &= C^{\infty}(M, M_3(\mathbb{C}));\nonumber\\
 \H &= L^2(M, S)\otimes \left( \C^3 \oplus M_3(\mathbb{C}) \oplus \C^3 \right);\nonumber\\
     D &= \slashed{\partial}_M\otimes \id + \gamma_5 \otimes D_F.
\end{align*}
The grading $\gamma = \gamma_5 \otimes 1$ and real structure $J= J_M \otimes J_F$ give the resulting real spectral triple KO-dimension 4. 

We will write a generic element in the Hilbert space as $\psi =(\psi_q , \psi_g , \psi_{\bar q})$, according to the above direct sum decomposition. For the bosons, we derive from Equation \eqref{eq:innerfluctuationform} and Lemma \ref{lma:inner-finite} that
\begin{prop}
The inner fluctuations $D \mapsto D_A = D + A + J A J^*$ with $A \in \Omega^1_D(\A)$ are parametrized by an $SU(3)$-gauge potential $\bA_\mu(x)$ ($\mu =1 ,\ldots, 4$) and a $\C^3$-valued function $\squark(x)$ ($x \in M$). Explicitly, we have with $\bA = i \gamma^\mu \bA_\mu$:
$$
D_A = \dirac \otimes 1 + \bA + \gamma_5 \otimes D_{\tilde q}
$$
with $D_{\tilde q}$ as in Lemma \ref{lma:inner-finite}.
\end{prop}
We will identify $\squark$ and $\antisquark$ as the \emph{squark} and \emph{anti-squark}, respectively. As before, $\bA$ will be the gluon, and $\psi_g$ the gluino. This terminology is justified by the action of the gauge group on these fields:
\begin{prop} The gauge group $SU(\A) = C^\infty(M,SU(3))$ acts on the squarks and quarks in the defining representation, on the gluinos in the adjoint representation, and on the gluon as a $SU(3)$-gauge field, {i.e.} for $u \in SU(\A)$:
\begin{align*}
\tilde q \mapsto u \tilde q; \qquad \psi_q \mapsto u \psi_q; \qquad \psi_g \mapsto u \psi_g u^*;\qquad \bA_\mu \mapsto u \bA_\mu u^* + u \partial_\mu u^*.
\end{align*}
    \end{prop}
    \begin{proof}
For a real spectral triple, the gauge group $SU(\A)$ acts on the Hilbert space in the adjoint representation $\Ad(u) := \pi(u) \pi^\circ(u^*) = u J u J^*$.
A direct computation shows that on an element in Hilbert space:
      \begin{align}
        \Ad(u)(\psi_q, \psi_g, \psi_{\bar q}) = (u \psi_q, u \psi_g u^*, \overline{u}\psi_{\bar q})\nonumber.
      \end{align}
      Next, we look at how $D_{\squark}$ transforms:
      \begin{align}
        \Ad(u) D_{\squark}\Ad(u^*)(\psi_q, \psi_g, \psi_{\bar q})
        &= \Ad(u) D_{\squark}(u^*\psi_q, u^*\psi_gu, u^t\psi_{\bar q})\nonumber\\
        &= \Ad(u)\big((u^*\psi_gu)\squark, (u^*\psi_q)\antisquark^{\,t} + \squark (u^t\psi_{\bar q})^t, (u^*\psi_gu)^t\antisquark\big)\nonumber\\
        &= \big(u(u^*\psi_gu)\squark, u(u^*\psi_q)\antisquark^{\,t}u^* + u\squark (u^t\psi_{\bar q})^tu^*, \overline{u}(u^*\psi_gu)^t\antisquark\big)\nonumber\\
        &= \big(\psi_gu\squark, \psi_q(\overline{u}\antisquark)^{\,t} + u\squark \psi_{\bar q}^t, \psi_g^t\overline{u}\antisquark\big)\nonumber
      \end{align}
      which corresponds to applying $D_{u\squark}$ to $(\psi_q, \psi_g, \psi_{\bar q})$. 
      
      Last, we check how the gluons transform. For instance, when applied to a gluino $\psi_g$: 
      \begin{align}
        \Ad(u) (\partial_\mu + \mathbb{A}_\mu) \Ad(u^*)\psi_g &= \pi(u)\pi^\circ(u^*)(\partial_\mu(u^*\psi_gu) + [A_\mu, u^*\psi_gu])\nonumber\\
        &= \Ad(u) \psi_g + \partial_\mu \psi_g + \psi_g(\partial_\mu u)u^* + u[A_\mu, u^*\psi_gu]u^*\nonumber\\
        &= \partial_\mu \psi_g + \ad(uA_\mu u^* + u[\partial_\mu, u^*])\psi_g\nonumber\\
        &= (\partial_\mu + \mathbb{A}^u_\mu)\psi_g\nonumber,
      \end{align}
      with $A^u$ as in \eqref{eq:A_gauge_trans}. Similar statements hold when acting on $\psi_q$ and $\psi_{\bar q}$, respectively.
    \end{proof}

\subsection{The spectral action}
  Having found an expression for the inner fluctuations of the product geometry, we now determine the corresponding spectral and fermionic action. Let us abbreviate $D^{(1,0)} = \dirac \otimes 1 + \bA \equiv i \gamma^\mu D_\mu$ to write for the fluctuated Dirac operator:
  \begin{align}
    D_A = D^{(1,0)} + \gamma^5\otimes D_{\squark}\label{eq:full_dirac},
  \end{align}
Before we compute the spectral action, we will first prove some useful lemmas.
    
    \begin{lem}\label{lem:DFsquared}
    For the square of $D_A$ we have
    \begin{align*}
       \tr D_A^2 = \tr (D^{(1,0)})^2 + \tr (D_{\squark}^2)
    \end{align*}
    with
     \begin{align}
      \tr (D_{\squark}^2) &= 12g_3^2|\squark|^2\nonumber
    \end{align}
    \end{lem}
    \begin{proof}
The cross term in the square of $D_A$ equals $\gamma^5[D^{(1,0)}, 1\otimes D_{\squark}]$, which vanishes upon taking the trace. For the square of the finite part we find 
    \begin{align}
      \tr (D_{\squark}^2) &= 2g_3^2\tr[A^*_{\squark}A_{\squark}] + 2g_3^2\tr[B^*_{\antisquark}B_{\antisquark}]\nonumber.
    \end{align}
   If we apply this first operator on the right hand side on a quark $q$ we find that    
    $A_{\squark}A^*_{\squark} = \diag|\squark|^2$. With a similar calculation for $B_{\antisquark}$, we arrive at the result.
    \end{proof}
    
    \begin{lem}\label{lem:DFfourth}
      For the fourth power of the finite Dirac operator $D_{\squark}$ we have
      \begin{align}
          \tr D_{\squark}^4 = 16g_3^4|\squark|^4.
      \end{align}
    \end{lem}
    \begin{proof}
       The calculation bears strong resemblance with the previous lemma, the main difference lies in additional cross terms. 
     \end{proof}  
  
   \begin{lem}\label{lem:DFcomm} For the commutator between the continuous $D^{(1,0)} = i \gamma^\mu D_\mu$ and finite Dirac operators we have
      \begin{align}
        [D_\mu, D_{\squark}](\psi_q, g, \psi_{\bar q}) &= D_{(\partial_{\mu} + g_3A_\mu)\squark}(\psi_q, g, \psi_{\bar q})\nonumber.
      \end{align}
  \end{lem}
  \begin{proof}
     We use that $D_\mu$ acts on the Hilbert space as:
     \begin{align*}
D_\mu (\psi_q, \psi_g, \psi_{\bar q}) = \left( (\partial_\mu + g_3A_\mu)\psi_q, (\partial_\mu + g_3\mathbb{A}_\mu)\psi_g, (\partial_\mu + g_3\overline{A}_\mu)\psi_{\bar q} \right)
      \end{align*}
      Thus we get  
from applying the commutator (whilst putting $g_3=1$ for simplicity):
      \begin{align}
         [D_\mu, D_{\tilde q}](\psi_q, g, \psi_{\bar q})
         &= \big(g(\partial_\mu + A_\mu)\squark, \psi_q[\partial_\mu\antisquark]^t - (\psi_q\antisquark^{\,t}) A_\mu + [(\partial_\mu + A_\mu)\squark]\psi_{\bar q}^{t}, g^t(\partial_\mu - A_\mu^t)\antisquark\big)\nonumber\\
         &= \big(g(\partial_\mu + A_\mu)\squark, \psi_q[\overline{(\partial_\mu + A_\mu)\squark}]^{t} + [(\partial_\mu + A_\mu )\squark] \psi_{\bar q}^{t}, g^t\overline{(\partial_\mu + A_\mu)\squark}\big),\nonumber
      \end{align}
      where we have frequently used that $A^*_\mu = -A_\mu$.
  \end{proof}
      
We will proceed ---as in Section \ref{ch:eym}--- by making an expansion in powers of $D_A^2$. We first determine the endomorphism $E'$ defined by 
    \begin{align*}
      D_A^2 = \nabla^*\nabla - E'.
    \end{align*}
Here $\nabla$ 
is the connection defined by $\bA$ on the tensor product of the spinor bundle by the trivial bundle with fiber $\C^3 \oplus M_3(\C) \oplus \C^3$. With respect to the Einstein--Yang--Mills system, we are simply adding the term $\gamma_5 \otimes D_{\tilde q}$ to $D^{(1,0)}$; this is easily seen to leave $\Omega_{\mu\nu}$ unchanged, and having the following effect on $E$:
    \begin{align}
      - E &\mapsto -E' = -E - i\gamma^5\gamma^\mu[D_\mu, D_{\tilde q}] +  D_{\tilde q}^2
\label{eq:neqE}
    \end{align}
compared to $E = \frac{1}{4}R\otimes\id - \frac{1}{2}\gamma^\mu\gamma^\nu\otimes\mathbb{F}_{\mu\nu}$ prior to adding squarks and quarks. The minus sign giving rise to the commutator comes from interchanging $\gamma^\mu$ and $\gamma^5$. The term $\omega_\mu$ then drops from the expression, leaving the commutator of $D_\mu := \partial_\mu + g_3\mathbb{A}_\mu$ with $D_{\tilde q}$.

    \begin{thm}
The spectral action $S_b[A]$ for the inner fluctuations of the spectral triple $(\A,\H,D)$ introduced above is given by the spectral action $S'_{b}$ for the Einstein--Yang--Mills system (cf. Theorem \ref{thm:bos-action-EYM}) plus additional terms of the form
      \begin{align*}
         S_b[A] &=S'_{b}[A] + \int_M\left[-\frac{6f_2}{\pi^2}g_3^2\Lambda^2|\squark(x)|^2 + g_3^2\frac{f(0)}{4\pi^2}(8g_3^2|\squark(x)|^4
+ 6|D_{\mu}\squark(x)|^2 - 3Rg_3^2|\squark(x)|^2)\right]\volelt
      \end{align*}
\label{thm:bos-action-sQCD}
    \end{thm}
    \begin{proof}
      From \eqref{eq:gilkey2} wee see that the contributions to the Lagrangian of $\mathcal{O}(\Lambda^2)$ come from $\tr(E')$. Since the trace of the second term of \eqref{eq:neqE} vanishes, we are left with
      \begin{align}
        \tr(E') &= \tr(E) - 4\tr (D^{(0,1)})^2 = \tr(E) - 48g_3^2|\squark|^2\nonumber, 
      \end{align}
      by virtue of Lemma \ref{lem:DFsquared}. \\
      
      Since $\Omega_{\mu\nu}$ is unaltered, all extra terms we have on $\mathcal{O}(\Lambda^0)$ result from $\tr(RE')$ and $\tr(E'^2)$ [see \eqref{eq:gilkey3}]. For the first we have
      \begin{align}
        \tr(RE') = \tr(RE) - 4R\tr(D^{(0,1)})^2 = \tr(RE) - 48g_3^2R|\squark|^2,
      \end{align}
      whereas the second gives
      \begin{align}
        \tr(E'^2) &= \tr(E^2) + \tr(i\gamma^5\gamma^\mu[D_\mu, D_{\squark}])^2 + \tr[(D_{\tilde q})^2]^2 -\frac{1}{2}\tr[ R\otimes (D_{\tilde q})^2] \nonumber\\
        &= \tr(E^2) + 4\tr([D_\mu, D_{\tilde q}][D^\mu, D_{\tilde q}]) + 4\tr(D_{\tilde q}^4)- 2R\tr(D_{\tilde q}^2)\label{eq:traceEnewsquared}.
      \end{align}
In the first step we have used that terms of the Clifford algebra proportional to $\gamma^{\mu}\gamma^{\nu}$ ($\mu < \nu$), $\gamma^5\gamma^\mu$ and $1$ are orthogonal, and we consequently only retain the squares of the terms in \eqref{eq:neqE} plus one cross-term. Now for the second and the last terms of \eqref{eq:traceEnewsquared} we can use Lemmas \ref{lem:DFfourth} and \ref{lem:DFcomm} with which the former becomes
      \begin{align}
        \tr([D_\mu, D_{\tilde q}][D^\mu, D_{\tilde q}]) &= \tr D_{(\partial_\mu + g_3A_\mu)\squark}D_{(\partial_{\mu} + g_3A_\mu)\squark}= 12g_3^2|(\partial_{\mu} + g_3A_\mu)\squark|^2 \label{eq:squark_kinetic}.
      \end{align}
      Taking the expansion of the spectral action \eqref{eq:expansion_action_functional}, with the coefficients taken from \eqref{eq:gilkey2} and \eqref{eq:gilkey3} we get the following extra extra contributions:
      \begin{align}
        \text{order }\Lambda^2&:\qquad -2f_2 \frac{1}{(4\pi)^2} 4\tr(D_{\tilde q})^2 = -\frac{6}{\pi^2}f_2g_3^2|\squark|^2,\nonumber\\
        \text{order }\Lambda^0&:\qquad f(0) \frac{1}{(4\pi)^2}\frac{1}{360}\big[-60(- 48g_3^2R|\squark|^2) + 180(4\cdot 12|(\partial_\mu + g_3A_\mu)\squark|^2+ 64|\squark|^4 -24R|\squark|^2)\big]\nonumber
       \end{align} 
      which ends the proof.
    \end{proof}

In order to have manifest supersymmetry with the fermionic action $S_f[A,\psi_q,\psi_g]$ we have to reduce once more the degrees of freedom for the spinor $\psi_g$. This is completely analogous to what happens in the Einstein--Yang--Mills system in Section \ref{ch:susy_in_EYM}: we replace the $M_3(\C)$-valued Dirac spinor $\psi_g$ by two $\su(3)$-valued Weyl spinors $\psi_g$ and $\chi_g$ of opposite chirality. 
\begin{thm}
The fermionic action for the triple $(\A,\H,D)$ is given by 
\begin{align*}
S_f[A,\psi_q,  \psi_g,\chi_g, \psi_{\overline{q}}]&\equiv \left\langle (\psi_q, \chi_g, \psi_{\bar q}), D_A  (\psi_q, \psi_g, \psi_{\bar q})\right\rangle\\
 &=\langle \psi_q, (\dirac + A) \psi_q \rangle + \langle\chi_g, (\dirac+\bA) \psi_g \rangle+ \langle \psi_{\bar q}, (\dirac+\bar A) \psi_{\bar q}\rangle \\
& \qquad+\langle \psi_q, \psi_g \tilde q \rangle + \langle \chi_g \tilde q, \psi_q \rangle + \langle \chi_g^t \bar{\tilde q}, \psi_{\overline{q}} \rangle + \langle \psi_{\overline{q}}, \psi_g^t \bar{\tilde q} \rangle
\end{align*}
\label{thm:ferm-action-sQCD}
\end{thm}

\subsection{Discussion}

In summary, we have added the superpartners of the QCD-particles (squarks and gluinos) to the theory, in conformity to the `paradigm' of NCG: fermions are elements of the Hilbert space, whereas bosons arise as inner fluctuations of a Dirac operator. Having quarks, gluinos and anti-quarks as the fermionic constituents, the freedom to choose the (finite part of the) Dirac operator was seen to be very little. On top of that, this construction led to the fact that these superpartners are in the right representation of the gauge group. A computation of the spectral action and the fermionic action then led to additional terms over the supersymmetric Einstein--Yang--Mills system considered in Section \ref{ch:susy_in_EYM}. We will now interpret these additional terms.

Note that for supersymmetry at least the number of degrees of freedom need to be the same. For that, the finite part of the gluinos has to be reduced from $M_3(\mathbb{C})$ to $\mathfrak{su}(3)$ ---a problem that was dealt with in Section \ref{ch:susy_in_EYM}. As far as the quarks and squarks are concerned, we have not addressed the apparent discrepancy between degrees of freedom yet. Indeed, the squarks are described by a $\C^3$-valued function, whereas quarks are described by a $\C^3$-valued Dirac spinor, {i.e.} a mismatch of a factor of 4. This is due to the fact that we have ignored isospin, something that will await another time.

We next compare the above results (Proposition \ref{thm:bos-action-sQCD} and \ref{thm:ferm-action-sQCD}) with that of the Minimally Supersymmetric Standard Model (MSSM); Kraml \cite{Krm99} and Chung et al.~ \cite{CEKKLW05} provide lengthy expositions on the MSSM. In the latter, the various MSSM-interactions are conveniently listed in the appendix.
We first switch to flat Euclidean space by taking $\omega_\mu = 0$ and $R = 0$ and working on a local chart of $M$. For each of the interactions that appear we will at the same time make the switch from the current notation to the one more common in physics and translate the (relevant pieces of the) Lagrangian as found in the literature to this context.

First, the free part of the action $S_b+S_f$ in Proposition \ref{thm:bos-action-sQCD} and \ref{thm:ferm-action-sQCD} coincides with the usual kinetic terms for the quark, squark, gluon and gluino. Note the additional coupling of the squark to the scalar curvature of $M$.

\begin{itemize}
\item \emph{Squark-quark-gluino}

The quark is described by $\psi_q= \psi_q^i \otimes e_i \in L^2(M,S) \otimes \C^3$, the antiquark by $\psi_{\bar q}= \psi_{\bar q}^i \otimes e_i \in L^2(M,S )\otimes \C^3$, and the gluino by a pair of $\su(3)$-valued Weyl spinors $\psi_g =\psi_g^a \otimes T_a$ and $\chi_g = \chi_g^a \otimes T_a$. The finite part of the Dirac operator gives in the fermionic action the term:
  \begin{align*}
&   \left( (\psi_q, \chi_g, \psi_{\bar q}),  (\gamma^5\otimes D_{\tilde q}) ( \psi_q, \psi_g. \psi_{\bar q} ) \right)\\
& \qquad = 
\left( (\psi_q, \chi_g, \psi_{\bar q}),  (\gamma_5 \psi_g \squark, \gamma_5 (\psi_q \squark^t + \squark \psi_{\bar q}^t), \gamma_5 \psi_g^t \antisquark ) \right)
\\
    & \qquad= g_3(T_a)_{ik}\bigg[\rinpr{\psi_{q}^i}{\gamma^5\psi^a_{g}}\squark^k - \rinpr{\chi_g^a}{\gamma^5\psi^k_q}\antisquark^i -  \rinpr{\chi_g^a}{\gamma^5\psi^i_{\bar q}}\squark^k + \rinpr{\psi_{\bar q}^k}{\gamma_5 \psi_g^a} \antisquark^i   \bigg]\nonumber
  \end{align*}
Here the transpose ${}^t$ refers to the finite index only and $(\cdot,\cdot)$ is the hermitian structure in the spinor bundle ({i.e.} summation over spinor indices). Note that an interaction such as $( \psi_q^i , \gamma^5 \psi_g^a)$ actually only involves the positive chirality part (with respect to $\gamma_5$) of $\psi_q^i$ in accordance with \cite{Krm99,CEKKLW05}.

\item  As in the Einstein--Yang--Mills system, we get a \emph{gluon-gluino-gluino} interaction from the continuous part of $D_A$ in the fermionic action:
\begin{align*}
\rinpr{\chi_g}{ig_3\gamma^\mu\mathbb{A}_\mu \psi_g}=
  ig_3\rinpr{\chi^c_g}{\gamma^\mu A^a_\mu\psi^b_g} \tr \left(T_c[T_a, T_b]\right) =  ig_3f_{abc}\rinpr{\chi^c_g}{\gamma^\mu \psi^b_g }A^a_\mu
\end{align*}
Similarly, we have the usual {\it quark-quark-gluon} interaction, which reads $\rinpr{\psi_q}{ ig_3 \gamma^\mu A_\mu \psi_q}$.

\item  From \eqref{eq:squark_kinetic} we can extract a \emph{squark-squark-gluon} term, that is of the form
\begin{align}
  &- g^2_3(g_3A_\mu\squark)_i(\partial^\mu\antisquark)^i - g_3^2(\partial_\mu\squark)^i(g_3\overline{A_\mu}\antisquark)_i\nonumber\\
  &\qquad= - g_3^3A^a_\mu (T_a)_{ij}\squark^j\partial^\mu(\antisquark)^i - g_3^3\partial_\mu(\squark)^iA^a_\mu (\overline{T_a})_{ij}\antisquark^{\,j}\nonumber\\
  &\qquad
=  g_3^3A^a_\mu (T_a)_{ij}[\antisquark^{\,i}(\partial_\mu\squark)^j - g_3^3(\partial^\mu\antisquark)^i\squark^j]\nonumber  
\end{align}

\item Equation \eqref{eq:squark_kinetic} provides us a \emph{squark-squark-gluon-gluon} term as well: 
\begin{align}
  g_3^2(g_3A^\mu\squark)^i(g_3\overline{A_\mu\squark})_i 
&= -g_3^4A^{\mu\,a}A_\mu^b(T_bT_a)_{ij}\antisquark_i\squark_j\nonumber\\
&=-\frac{1}{6}A^{\mu}_{\phantom{\mu}a}A_\mu^{\phantom{\mu}a}\antisquark_i\squark^i - 
 \frac{1}{2}d_{abc}\,A^{\mu\,a}A_\mu^b(T^c)_{ij}\antisquark_i\squark_j\nonumber\end{align}     
In going to the last line, we have use the identity
\begin{align}
 T_{b}T_{a} = \frac{1}{6}\delta_{ab}\id_3 + \frac{1}{2}(if_{bac} + d_{bac})T^c\nonumber,
\end{align}
where the term with $f_{abc}$ vanishes since $A^{\mu\,a}A_\mu^b$ is symmetric upon interchanging $a$ and $b$.\\

\item Finally, there is a \emph{four squark self-interaction} 
\begin{align}
  g_3^4|\squark(x)|^4 = g_3^4\squark(x)_i\antisquark(x)^i\squark(x)_j\antisquark(x)^j\nonumber,  
\end{align}
originating from the third term of the display in Theorem in \ref{thm:bos-action-sQCD}.
 
\end{itemize}

To summarize, all results are in perfect agreement with the literature, in the sense that all interactions are present and their form is precisely the same. In three terms that we compared however, we were off by two powers of the coupling constants and a sign. However, it are precisely these `erroneous' terms of the Lagrangian that are accompanied by a factor $f(0)$, in which we can absorb this excess of coupling constants. The minus sign is unresolved still, since $f$ has to be a positive function. There is one other unresolved issue: the constants appearing in our results do not in all cases match those of the literature. 
However, to properly address all these issues, one has to wait for a description of the full MSSM in terms of a noncommutative manifold ---since that is what we are comparing our model with here--- taking also into account isospin and hypercharge. This is part of future research.

One observation that we cannot refrain from doing is that the sum $S_b+S_f$ of the actions in Theorem \ref{thm:bos-action-sQCD} and \ref{thm:ferm-action-sQCD} is not supersymmetric. In fact, there appear squark mass terms as allowed in soft supersymmetry breaking (see for instance \cite{CEKKLW05} and references therein). We consider the presence of these terms as a merit of the above model, leaving the question open whether a description of the  spontaneous supersymmetry breaking mechanism responsible for these soft-breaking terms can be found within noncommutative geometry. Of course, a search for such a mechanism is motivated by the derivation of the Higgs spontaneous {\it gauge} symmetry breaking mechanism from a noncommutative manifold in the case of the Standard Model. Possibly, one of the noncommutative manifolds that appear in the classification of \cite{CC07b} will describe the supersymmetric theory with spontaneous supersymmetry breaking mechanism.

\appendix

\section{Fierz identities}\label{se:fiers_trans}
The topic of this section are (Euclidean) Fierz identities. For these identities in a Minkowskian background, we refer to eg. \cite{NP04}.

\begin{defin}[Orthonormal Clifford basis] Let $Cl(V)$ be the Clifford algebra over a vector space $V$ of dimension $n$. Then $\gamma_K := \gamma_{k_1}\cdots\gamma_{k_r}$ for all strictly ordered sets $K = \{k_1 < \ldots < k_r\} \subseteq \{1, \ldots, n\}$ form a basis for $Cl(V)$. If $\gamma_K$ is as above, we denote with $\gamma^K$ the element $\gamma^{k_1}\cdots\gamma^{k_r}$. The basis spanned by the $\gamma_K$ is said to be \textit{orthonormal} if $\tr\gamma_K\gamma_L = nn_K\delta_{KL}\ \forall\ K, L$. Here $n_K := (-1)^{r(r-1)/2}$, where $r$ denotes the cardinality of the set $K$ and with $\delta_{KL}$ we mean 
\begin{align}
  \delta_{KL} = \left\{\begin{array}{ll} 1\quad \text{if}\ K = L\\
                         0 \quad \text{else}\\
                       \end{array}.
                \right.
\end{align}
\end{defin}

\begin{exmpl}\label{exmpl:dim4} Take $V = \mathbb{R}^4$ en let $Cl(4, 0)$ be the Euclidean Clifford algebra [i.e. with signature \mbox{($+$\ $+$\ $+$\ $+$)}]. Its basis are the sixteen matrices
\begin{align}
 & 1                                  &                     &\quad \nonumber\\
 & \gamma_\mu                         &                     &\quad \text{(4 elements)}\nonumber\\
 & \gamma_\mu\gamma_\nu\quad          & \mu < \nu           &\quad \text{(6 elements)},\nonumber\\
 & \gamma_\mu\gamma_\nu\gamma_\lambda & \mu < \nu < \lambda &\quad \text{(4 elements)}\nonumber\\
 &\gamma_1\gamma_2\gamma_3\gamma_4=:\gamma_5.    &           &\nonumber
\end{align}
We can identify
\begin{equation}\label{eq:mink_eucl}
\begin{split}
 \gamma_1\gamma_2\gamma_3 &= \gamma_4\gamma_5,\qquad \gamma_1\gamma_3\gamma_4 = \gamma_2\gamma_5\\ 
 \gamma_1\gamma_2\gamma_4 &= - \gamma_3\gamma_5,\qquad \gamma_2\gamma_3\gamma_4 = -\gamma_1\gamma_5,
\end{split} 
\end{equation}
establishing a connection with the basis most commonly used by physicists.
\end{exmpl}

\begin{lem}[Completeness relation] If the basis of the Clifford algebra is orthonormal, it satisfies the following completeness relation:
\begin{align}
	\frac{1}{n}\sum_L n_L(\gamma^L)_{d}^{\phantom{d}c}(\gamma_L)_{a}^{\phantom{a}b} &= \delta_a^{\phantom{a}c}\delta_d^{\phantom{d}b}\label{eq:completeness}.
\end{align}
\end{lem}
\begin{proof}
Since the $\gamma_K$ form a basis, we can write any element $\Gamma$ of the Clifford algebra as 
\begin{align}
	\Gamma &= \sum_K m^K \gamma_{K}\quad m_K \in \mathbb{C}\label{eq:completeness2},
\end{align}	
where the sum runs over all (strictly ordered) sets. By multiplying both sides with $\gamma^L$ and taking the trace we find the expression for the coefficient $m^L$ to be:
\begin{align}
	m^L &= \frac{1}{n}n_L\tr \Gamma\gamma^L\nonumber, 
\end{align}
Applying this result in particular to $\Gamma = \gamma_K$, and writing matrix indices explicitly, \eqref{eq:completeness2} yields
\begin{align}
	(\gamma_K)_a^{\phantom{a}b} &= \frac{1}{n}\sum_L n_L(\gamma_K)_c^{\phantom{c}d}(\gamma^L)_{d}^{\phantom{d}c}(\gamma_L)_{a}^{\phantom{a}b}\nonumber,
\end{align}
for which \eqref{eq:completeness} is required.
\end{proof}

\begin{thm}[(Generalized) Fierz identity]\label{thm:fierz} If for any two strictly ordered sets $K, L$ there exists a third strictly ordered set $M$ and $c \in \mathbb{N}$ such that $\gamma_K\gamma_L = c\,\gamma_M$, we have the four-spinor identity
\begin{align}
   \inpr{\psi_1}{\gamma^K\psi_2}\inpr{\psi_3}{\gamma_K\psi_4} &= -\frac{1}{n}\sum_L C_{KL}\inpr{\psi_3}{\gamma^L\psi_2}\inpr{\psi_1}{\gamma_{L}\psi_4}\nonumber\\
   &\qquad\qquad\qquad C_{KL} \in \mathbb{N}\label{eq:fierzf},
\end{align}
for any $\psi_1, \ldots, \psi_4$ in the $n$-dimensional spin representation of the Clifford algebra. Here we denote by $\inpr{.}{.}$ the inner product on the spinor representation.
\end{thm}

\begin{proof}
	We start by multiplying the completeness relation \eqref{eq:completeness} with $(\gamma_K)_{c}^{\phantom{c}e}(\gamma^K)_{b}^{\phantom{b}f}$ yielding
	\begin{align}
		 (\gamma_K)_{a}^{\phantom{a}e}(\gamma^K)_{d}^{\phantom{d}f} &= \frac{1}{n}\sum_L (\gamma_L\gamma_K)_{d}^{\phantom{d}e}(\gamma^L\gamma^K)_{a}^{\phantom{a}f}\nonumber,
	\end{align}
	or
	\begin{align}
	  (\gamma_K)_{a}^{\phantom{a}e}(\gamma^K)_{d}^{\phantom{d}f} &= \frac{1}{n}\sum_M C_{KM}(\gamma_M)_{d}^{\phantom{d}e}(\gamma^M)_{a}^{\phantom{a}f}\label{eq:fierz_interm},
	\end{align}
	by the assumption made. Here we have accommodated the proportionality constants in a matrix $C_{KL}$. Now we have to contract the above expression with the four spinors $\overline{\psi_1^a}, {\psi_2}_e, \overline{\psi_3^d}$ and ${\psi_4}_f$. But, remembering that they are Grassmann variables --- i.e.~their components anticommute--- we get one minus sign on the left hand side of \eqref{eq:fierz_interm} from interchanging $\psi_1$ and $\psi_3$. Hence we we arrive at the result. 
\end{proof}
Now how do we compute the constants $C_{KL}$? Just multiply \eqref{eq:fierz_interm} again by $(\gamma_L)_{e}^{\phantom{e}d} (\gamma^L)_{f}^{\phantom{f}a}$, yielding:
	\begin{align}
	 \tr(\gamma^K\gamma^L\gamma_K\gamma_L) &= \frac{1}{n}\sum_M C_{KM}\tr(\gamma^M\gamma_L)\tr(\gamma_M\gamma^L) \label{eq:valuec}.
	\end{align}
	On the other hand, we have
	\begin{align}
	  \gamma^K\gamma^L\gamma_K = f_{KL}\gamma^L\quad f_{KL} \in \mathbb{N}\ \ \text{(no sum over}\ L\text{)}\label{eq:deffmatrix}
	\end{align}	
	using the anticommutator repeatedly\footnote{For example: $\gamma^\mu\gamma^\lambda\gamma_\mu = (2-\dim V)\gamma^\lambda\quad\forall\ \lambda \in \{1, 2, \cdots, \dim V\}$.}. Putting \eqref{eq:deffmatrix} into \eqref{eq:valuec} we get:
	\begin{align}
	 f_{KL}\tr(\gamma^L\gamma_L) &= n\sum_M C_{KM} \delta^M_L\delta_M^L\nonumber
	\end{align}
	or
	\begin{align}
	  C_{KL} &= n_Lf_{KL} \label{eq:valuec2},
	\end{align}
	since
	\begin{align}
	 \tr(\gamma^L\gamma_L) &= 
			(-1)^{r(r-1)/2}n\nonumber,
	\end{align}
	by orthonormality.

\begin{cor}[Fierz identity] We work out one example of particular interest to us. Consider again $Cl(4, 0)$ ($n=4$) with the basis as in Example \ref{exmpl:dim4}. As can readily be checked, this basis satisfies the requirement for theorem \ref{thm:fierz}. The spinors we will contract with, are the four Weyl spinors: $\chi, \epsilon_- \in \mathcal{S}^-, \psi_1, \psi_2 \in \mathcal{S}^+$. We start with determining the numbers $f_{1r}, r = 0, \ldots, 4$ defined by $\gamma^\mu\gamma_L\gamma_\mu = f_{1r}\gamma_L$ (see above) where $r$ is the cardinality of $L$. We find the recursive relation 
\begin{align}
  \gamma^\mu 1 \gamma_\mu &= n\cdot 1 \equiv f_{10} 1\nonumber\\
  \gamma^\mu \gamma_\nu \gamma_\mu &= 2\gamma_\nu - \gamma^\mu\gamma_\mu\gamma_\nu = (2 - f_{10})\gamma^\nu \equiv f_{11}\gamma^\nu\nonumber\\
  & \ldots \nonumber\\
  \gamma^\mu(\gamma^{\nu_1}\cdots\gamma^{\nu_n})\gamma_\mu &= [2(-1)^{n-1} - f_{1(n-1)}]\gamma^{\nu_1}\cdots\gamma^{\nu_n}\nonumber\\
  &\equiv f_{1n}\gamma^{\nu_1}\cdots\gamma^{\nu_n}\qquad(n \leq 4)\nonumber
\end{align}
which gives
\begin{align}
  f_{10} = 4,\quad f_{11} = -2,\quad f_{12} = 0, \quad f_{13} = 2, \quad f_{14} = -4\nonumber
\end{align}
and consequently, using \eqref{eq:valuec2}
\begin{align}
  C_{10} = 4,\quad C_{11} = -2,\quad C_{12} = 0, \quad C_{13} = -2, \quad C_{14} = -4\nonumber.
\end{align}
Now applying \eqref{eq:fierzf} yields
\begin{align}
   \inpr{\chi}{\gamma^\mu\psi_1}\inpr{\epsilon_-}{\gamma_\mu\psi_2} &= -\frac{1}{4}C_{11}\inpr{\epsilon_-}{\gamma^\mu\psi_1}\inpr{\chi}{\gamma_\mu\psi_2}\nonumber\\
			&\quad- \frac{1}{4}C_{13}\inpr{\epsilon_-}{\gamma^\mu\gamma^\nu\gamma^\lambda\psi_1}\inpr{\chi}{\gamma_\mu\gamma_\nu\gamma_\lambda\psi_2}\nonumber
\end{align}
since only terms with an odd number of $\gamma$-matrices survive due to the different chirality of the spinors. Identifying the terms with three $\gamma$-matrices with $\pm\gamma^\mu\gamma_5$ as in \eqref{eq:mink_eucl}, we get
\begin{align}
	\inpr{\chi}{\gamma^\mu\psi_1}\inpr{\epsilon_-}{\gamma_\mu\psi_2} &= \frac{1}{2}\inpr{\epsilon_-}{\gamma^\mu\psi_1}\inpr{\chi}{\gamma_\mu\psi_2}\nonumber\\
			&\qquad +\frac{1}{2}\inpr{\epsilon_-}{\gamma^\mu\psi_1}\inpr{\chi}{\gamma_\mu\psi_2}\nonumber\\
			&=\inpr{\epsilon_-}{\gamma^\mu\psi_1}\inpr{\chi}{\gamma_\mu\psi_2}\label{eq:fierza}.
\end{align}
\end{cor}

\newcommand{\noopsort}[1]{}\def\cprime{$'$}


\begin{thebibliography}{10}

\bibitem{Cha94}
A.~H. Chamseddine.
\newblock {Connection between space-time supersymmetry and noncommutative
  geometry}.
\newblock {\em Phys. Lett.}, B332:349--357, 1994.

\bibitem{Cha98}
A.~H. Chamseddine.
\newblock {Remarks on the spectral action principle}.
\newblock {\em Phys. Lett.}, B436:84--90, 1998.

\bibitem{CC96}
A.~H. Chamseddine and A.~Connes.
\newblock Universal formula for noncommutative geometry actions: {U}nifications
  of gravity and the standard model.
\newblock {\em Phys. Rev. Lett.}, 77:4868--4871, 1996.

\bibitem{CC97}
A.~H. Chamseddine and A.~Connes.
\newblock The spectral action principle.
\newblock {\em Commun. Math. Phys.}, 186:731--750, 1997.

\bibitem{CC07a}
A.~H. Chamseddine and A.~Connes.
\newblock {Conceptual Explanation for the Algebra in the Noncommutative
  Approach to the Standard Model}.
\newblock {\em Phys. Rev. Lett.}, 99:191601, 2007.

\bibitem{CC07b}
A.~H. Chamseddine and A.~Connes.
\newblock {Why the Standard Model}.
\newblock {\em J. Geom. Phys.}, 58:38--47, 2008.

\bibitem{CCM07}
A.~H. Chamseddine, A.~Connes, and M.~Marcolli.
\newblock {Gravity and the standard model with neutrino mixing}.
\newblock {\em Adv. Theor. Math. Phys.}, 11:991--1089, 2007.

\bibitem{CEKKLW05}
D.~J.~H. Chung et~al.
\newblock {The soft supersymmetry-breaking Lagrangian: Theory and
  applications}.
\newblock {\em Phys. Rept.}, 407:1--203, 2005.

\bibitem{C94}
A.~Connes.
\newblock {\em Noncommutative Geometry}.
\newblock Academic Press, San Diego, 1994.

\bibitem{C96}
A.~Connes.
\newblock Gravity coupled with matter and the foundation of non-commutative
  geometry.
\newblock {\em Commun. Math. Phys.}, 182:155--176, 1996.

\bibitem{C08}
A.~Connes.
\newblock On the spectral characterization of manifolds.
\newblock preprint, 2009.

\bibitem{CL91}
A.~Connes and J.~Lott.
\newblock {Particle models and noncommutative geometry}.
\newblock {\em Nucl. Phys. Proc. Suppl.}, 18B:29--47, 1991.

\bibitem{CM07}
A.~Connes and M.~Marcolli.
\newblock {\em Noncommutative Geometry, Quantum Fields and Motives}.
\newblock AMS, Providence, 2008.

\bibitem{Gil84}
P.~B. Gilkey.
\newblock {\em Invariance theory, the heat equation, and the {A}tiyah-{S}inger
  index theorem}, volume~11 of {\em Mathematics Lecture Series}.
\newblock Publish or Perish Inc., Wilmington, DE, 1984.

\bibitem{HT91a}
F.~Hussain and G.~Thompson.
\newblock {Noncommutative geometry and supersymmetry}.
\newblock {\em Phys. Lett.}, B260:359--364, 1991.

\bibitem{HT91b}
F.~Hussain and G.~Thompson.
\newblock {Noncommutative geometry and supersymmetry. 2}.
\newblock {\em Phys. Lett.}, B265:307--310, 1991.

\bibitem{KW96}
W.~Kalau and M.~Walze.
\newblock {Supersymmetry and noncommutative geometry}.
\newblock {\em J. Geom. Phys.}, 22:77--102, 1997.

\bibitem{Krm99}
S.~Kraml.
\newblock {\em Stop and sbottom phenomenology in the MSSM}.
\newblock PhD thesis, TU Wien, 1999.
\newblock [hep-ph/9903257].

\bibitem{LMMS97}
F.~Lizzi, G.~Mangano, G.~Miele, and G.~Sparano.
\newblock {Fermion Hilbert space and fermion doubling in the noncommutative
  geometry approach to gauge theories}.
\newblock {\em Phys. Rev.}, D55:6357--6366, 1997.

\bibitem{Nic78}
H.~Nicolai.
\newblock {A Possible constructive approach to (super-{$\phi^3$})${}_4$. 1.
  Euclidean formulation of the model}.
\newblock {\em Nucl. Phys.}, B140:294, 1978.

\bibitem{NW96}
P.~{\noopsort{nieuwenhuizen}}van Nieuwenhuizen and A.~Waldron.
\newblock {On Euclidean spinors and Wick rotations}.
\newblock {\em Phys. Lett.}, B389:29--36, 1996.

\bibitem{NP04}
J.~F. Nieves and P.~B. Pal.
\newblock {Generalized Fierz identities}.
\newblock {\em Am. J. Phys.}, 72:1100--1108, 2004.

\bibitem{Rie74}
M.~A. Rieffel.
\newblock Morita equivalence for {$C\sp{\ast} $}-algebras and {$W\sp{\ast}
  $}-algebras.
\newblock {\em J. Pure Appl. Algebra}, 5:51--96, 1974.

\bibitem{Sch59}
J.~Schwinger.
\newblock {Euclidean Quantum Electrodynamics}.
\newblock {\em Phys. Rev.}, 115:721--731, 1959.

\bibitem{Sit08}
A.~Sitarz.
\newblock {Spectral action and neutrino mass}.
\newblock {\em Europhys. Lett.}, 86:10007, 2009.

\end{thebibliography}
\end{document}